\documentclass[11pt]{article}
\normalsize
\usepackage{amsfonts}
\usepackage{amsmath}
\usepackage{amssymb}
\usepackage{url}

\newtheorem{theorem}{Theorem}

\newtheorem{definition}{Definition}
\newtheorem{lemma}[theorem]{Lemma}
\newtheorem{corollary}[theorem]{Corollary}
\newtheorem{remark}{Remark}
\newtheorem{example}{Example}

\newenvironment{proof}{\noindent {\em Proof.}}{\hspace*{\fill} $\Box $\newline}

\title{An Algorithm for Computing the Covering Radius of a Linear Code Based on Vilenkin-Chrestenson Transform}

\author{Paskal Piperkov$^{b}$, Iliya Bouyukliev$^{b}$ and Stefka Bouyuklieva$^{a}$\\
$^{a}$ Faculty of Mathematics and Informatics,
Veliko Tarnovo University,\\ Veliko Tarnovo, Bulgaria,\\
$^{b}$ Institute of Mathematics and Informatics,\\
Bulgarian Academy of Sciences, Veliko Tarnovo, Bulgaria,\\
e-mail: iliyab@math.bas.bg; stefka@uni-vt.bg}

\date{}

\begin{document}

%
%
%


\maketitle

\begin{abstract}
We present a generalization of Walsh-Hadamard transform that is suitable for applications in Coding Theory, especially for computation of the weight distribution and the covering radius of a linear code over a finite field. The transform used in our research, is a modification of Vilenkin-Chrestenson transform. Instead of using all the vectors in the considered space, we take a maximal set of nonproportional vectors, which reduces the computational complexity.
\end{abstract}

%

\section{Introduction}

Minimum  distance, weight distribution and covering radius are some of the important parameters of the linear codes.  The minimum distance shows how many errors a particular code can correct and how many it can detect. When the maximum-likelihood decoding is performed the covering radius  is the measure of the largest number of errors in any correctable error pattern. Unfortunately, the problems for computing weight distribution, minimum distance and covering radius of a linear code are NP complete \cite{PBB:KaskiOstergard}. Algorithms that solve these problems are usually based on a search  in large sets of vectors \cite{PBB:Betten06,PBB:Cohen97}.

 Connections  between discrete Fourier type transforms, weight distribution and covering radius were established by Karpovsky \cite{PBB:Kar79,PBB:Kar81}. In his research, Karpovsky  considered  mainly binary codes and offered only ideas for the nonbinary  case. The size of the transformation vector in \cite{PBB:Kar81}  is $q^k$. In our study \cite{PBB:Bou20}, in which we calculate the weight distribution of a linear code, we succeeded to reduce the length of the transformed vector to $\theta (q, k) $, using an appropriate new transformation. In this way, the number of required calculations is reduced by about $q-1$ times.

 In this paper, we show that a similar transform can be used to calculate the covering radius of a linear code.
Section \ref{PBB:preliminaries} contains definitions and statements which are important in our research. In Section \ref{PBB:weight} we present a brief description of an algorithm for computing the weight distribution of a linear code based on an Walsh-Hadamard type transform. An algorithm for computing the covering radius of a linear code over a prime field is given in Section \ref{PBB:sec:covering}. The main results for linear codes over a composite field are proved in Section \ref{PBB:sec:coveringlarge}. We end the paper with a short conclusion.

\section{Preliminaries}
\label{PBB:preliminaries}

Let $\mathbb{F}_q^n$ be the $n$-dimensional vector space over the finite field $\mathbb{F}_q$, where $q$ is a prime power.

\subsection{Linear codes}

Every $k$-dimensional subspace $C$ of $\mathbb{F}_q^n$ is called a $q$-ary \textit{linear $[n,k]$
code} (or an $[n,k]_q$-code). The parameters $n$ and $k$ are called the \textit{length} and
\textit{dimension} of $C$, respectively, and the vectors in $C$ are called \textit
{codewords}. The \textit{(Hamming) weight} $\mbox{wt}(x)$ of a vector $x\in\mathbb{F}_q^n$
is the number of its non-zero coordinates. If $A_i$ is the number of
codewords of weight $i$ in $C$, $i=0,1,\dots, n$, then the sequence $(A_0, A_1,
\dots, A_n)$ is called the \textit{weight distribution} of $C$.
Any $k\times n$ matrix $G$, whose rows form a basis of $C$, is called a \textit{generator matrix}
of the code. An $(n-k)\times n$ matrix $H$, that determines the code $C$ in the sense that
\[C=\{x\in\mathbb{F}_q^n | Hx^T=\mathbf{0}\},\]
is called a \textit{parity check matrix} of the code. Note that the rows of $H$ are linearly independent.
For any vector $w\in \mathbb{F}_q^n$, the set $w+C=\{w+x|x\in C\}$ is called a \textit{coset} (or \textit{translate}) of the code. The \textit{weight of a coset} is the smallest weight of a vector in the coset, and any vector of this smallest weight in the coset is called
a \textit{coset leader}. The zero vector is the unique coset leader of the code $C$. The \textit{syndrome} of a vector $y\in\mathbb{F}_q^n$
 with respect to the parity check matrix $H$ is the vector $syn(y)=Hy^T\in\mathbb{F}_q^{n-k}$. Two vectors belong to the same coset if and only if they have the same syndrome \cite[Theorem 1.11.5]{PBB:Huf03}.
 The maximum integer among the weights of the cosets is called the \textit{the covering radius} of the code and denoted by $\rho(C)$. Moreover,
$\rho(C)$ is the smallest number $s$ such that every nonzero syndrome is a linear combination of $s$
or fewer columns of the parity check matrix $H$, and some syndrome requires $s$ columns \cite[Theorem 1.12.5]{PBB:Huf03}.
For more concepts and properties of linear codes we refer to \cite{PBB:Huf03,PBB:MWS}.

We consider only codes of full length, i.e. codes without zero columns in their generator matrices. If $C$ is a linear $[n,k]_q$ code of full length, and $\overline{C}$ is obtained from $C$ by adding $m$ zero coordinates to each codeword, then $\overline{C}$ is a linear code with length $n+m$, the same dimension $k$, the same weight distribution as $C$, and a covering radius $\rho(\overline{C})=\rho(C)+m$. Therefore it is enough to compute the weight distribution and the covering radius of the code $C$ in order to know these parameters for the code $\overline{C}$.

\subsection{Vilenkin-Chrestenson transform}

Let $\xi$ be a primitive complex $q$-th root of unity. We define the Vilenkin-Chrestenson
matrices of order $s$ by recurrence formulae as follows:
\begin{equation}
V_1=\left(\begin{array}{cllcl}
1&1&1&\ldots&1\\
1&\xi&\xi^2&\ldots&\xi^{q-1}\\
1&\xi^2&\xi^4&\ldots&\xi^{2(q-1)}\\
&\vdots&\vdots&&\vdots\\
1&\xi^{q-1}&\xi^{2(q-1)}&\ldots&\xi^{(q-1)^2}\end{array}\right),
\;V_{s+1}=V_1\otimes V_s, \; s\in \mathbb{Z}, \, s\geq 1,
\end{equation}
where $\otimes$ means the Kronecker product. We can consider the elements of the matrix $V_s$ in the following way:
\[V_s=\left(v_{\omega}(x)\right)_{\omega,x\in \mathbb{Z}_q^s},\]
where $\omega=(\omega_1,\ldots,\omega_s)$, $x=(x_1,\ldots,x_s)$, $\mathbb{Z}_q=\{ 0,1,\ldots,q-1\}$,
 $v_{\omega}(x)=\xi^{\omega\cdot x}$ and $\omega\cdot x=\sum_{i=1}^s \omega_ix_i\in \mathbb{Z}$.
 In what follows, we will use some properties of $v_{\omega}(x)$ that follow directly from the definition, namely
\begin{equation}\label{PBB:elementary}
 v_{\omega}(x)=v_{x}(\omega), \ \ \ v_{\omega}(x)v_{\omega}(y)=v_{\omega}(x+y), \ \ \ v_{\omega}(\mathbf{0})=1.
 \end{equation}

 The first property shows that the matrices $V_s$ are symmetric.

\begin{definition}
Let $h:\mathbb{Z}_q^s\to \mathbb{C}$ be a function. The \textit{Vilenkin-Chrestenson transform} of $h$ is a function $\widehat{h}:\mathbb{Z}_q^s\to \mathbb{C}$ defined by
\begin{equation}
\widehat{h}(\omega)=\sum_{x\in\mathbb{Z}_q^s}h(x)v_{\omega}(x),\quad \omega\in\mathbb{Z}_q^s.
\end{equation}
\end{definition}

Detailed information on this transform, as well as on other discrete transforms related to the Fourier transform, can be found in \cite{PBB:Bespalov,PBB:Farkov,PBB:spectral}.

Denote by $TT(h)$ the vector with the values of the function $h$ when the elements of $\mathbb{Z}_q^s$ are ordered lexicographically. This is an analog to the truth table of a Boolean function but here the coordinates of $TT(h)$ are complex numbers. The vectors of the function $h$ and its transform $\widehat{h}$ are connected by the equality
\[TT(\widehat{h})=V_s\cdot TT(h).\]
In this way we reduce Vilenkin-Chrestenson transform to a matrix by vector multiplication.

\section{Weight distribution of linear codes represented by a generator matrix}
\label{PBB:weight}

Let $\mathbb{F}_q=\{\alpha_0=0,\alpha_1=1,\alpha_2,\dots,\alpha_{q-1}\}$. For all positive integers $k$, we define the matrices $G_k$ recursively as follows:
\begin{equation}\label{PBB:Gk}
G_1=(1),\quad G_{k}=\left(\begin{array}{ccccc}\mathbf{0}&\mathbf{\alpha_1}&\ldots&\mathbf{\alpha_{q-1}}&1\\
G_{k-1}&G_{k-1}&\ldots&G_{k-1}&\mathbf{0}^T\end{array}\right),\; k\in\mathbb{Z}, k\geq 2,
\end{equation}
where $\mathbf{u}=(u,\dots,u)=u(1,1,\dots,1)=u.$\textbf{1}, $u\in\mathbb{F}_q$. The size of $G_k$ is $k\times \theta(q,k)$, where $\theta(q,k)=(q^k-1)/(q-1)$, and all columns of the matrix are pairwise linearly independent. Hence the vector-columns in $G_k$ form a maximal set of nonproportional vectors from the vector space $\mathbb{F}_q^k$. Since any such maximal set consists of the representatives of the points in the projective geometry $PG(k-1,q)$, we can say that the columns of the matrix $G_k$ represent all points in the projective geometry $PG(k-1,q)$.

The linear code, generated by the matrix $G_k$, is called a $q$-ary simplex code and denoted by $\mathcal{S}_{q,k}$. This code has length $\theta(q,k)$, dimension $k$ and weight distribution $A_0=1$, $A_{q^{k-1}}=q^k-1$, and $A_i=0$ for $i\neq q^{k-1}$, $1\le i\le \theta(q,k)$ (for more properties of the simplex codes see \cite{PBB:Huf03,PBB:MWS}).

Let $C$ be a linear $[n,k]_q$ code of full length with a generator matrix $G$.

\begin{definition}
The \textit{characteristic vector} of the code $C$ with respect to its generator matrix $G$ is the vector
\begin{equation}\label{PBB:chi}
\chi(C,G)=\left(\chi_1,\chi_2,\ldots,\chi_{\theta(q,k)}\right)\in \mathbb{Z}^{\theta(q,k)}
\end{equation}
where $\chi_u$ is the number of the columns of $G$ that are equal or proportional to the $u$-th column of $G_k$, $u=1,\ldots,\theta(q,k)$.
\end{definition}

When $C$ and $G$ are clear from the context, we write briefly $\chi$. Note that $\sum_{u=1}^{\theta(q,k)}\chi_u=n$, where $n$ is the length of $C$.

A code $C$ can have different characteristic vectors depending on the chosen generator matrices of $C$ and the considered generator matrix $G_k$ of the simplex code $\mathcal{S}_{q,k}$. If we permute the columns of the matrix $G$ and multiply them by nonzero elements of the field, we will obtain a monomially equivalent code to $C$ having the same characteristic vector. Moreover, from a characteristic vector one can restore the columns of the generator matrix $G$ but eventually at different order and/or multiplied by  nonzero elements of the field. This is not a problem for us because the equivalent codes have the same weight distributions.

Further, we consider the matrices $M_k=G_k^{\rm T}\cdot G_k$, $k\in \mathbb N$. We denote by ${\cal N}(M_k)$ the matrix obtained from $M_k$ by replacing all nonzero elements by $1$. The rows of the matrix $G_k^{\rm T}\cdot G$ represent a maximal set of codewords in $C$ that are pairwise nonproportional, and the Hamming weight of the $i$-th row of this matrix (multiplication over $\mathbb{F}_q$) is equal to the $i$-th coordinate of the column vector ${\cal N}(M_k)\cdot \chi^{\rm T}$ (multiplication over $\mathbb Z$), $i=1,\ldots,\theta(q,k)$. Therefore, the coordinates of ${\cal N}(M_k)\cdot \chi^{\rm T}$ provide sufficient information about the weight distribution of the code $C$ \cite[Lemma 1]{PBB:Bou20}.

If $m=(m_1,\ldots,m_{\theta(q,k)})$ is a row-vector in the matrix $M_k$, and $v=(v_1,\ldots,v_{\theta})\in\mathbb{Z}^{\theta}$ is a vector of length $\theta(q,k)$ with integer coordinates, we define the vector $m^{[v]}=(\mu_0,\mu_1,\ldots,\mu_{q-1})$
as follows
\[\mu_u=\sum_{\small{\begin{array}{c} j=1,\ldots,\theta(q,k)\\ \mbox{with}~m_j=\alpha_u\end{array}}} v_j \ \ \mbox{for each} \ u=0,\ldots,q-1.\]
In other words, $m^{[v]}$ can be computed from the vector
$$m'=(\underbrace{m_1,\ldots,m_1}_{v_1},\ldots,\underbrace{m_{\theta},\ldots,m_{\theta}}_{v_\theta}),$$
where there are $\mu_i$ occurrences of $\alpha_i$ in the vector $m'$. According to \cite[p. 142]{PBB:MWS}, $m^{[v]}=\mbox{comp}(m')$ is the composition of $m'$.
The matrix $M_k^{[v]}$ is obtained by replacing each row $m$ of $M_k$ by the corresponding vector $m^{[v]}$. Note that $M_k^{[v]}$ is a $\theta(q,k)\times q$ matrix. If we take $v$ to be the characteristic vector $\chi=\chi(C,G)$ of the linear code $C$ with a generator matrix $G$, then the $i$-th coordinate of
${\cal N}(M_k)\cdot\chi^T$ is equal to $n-\mu_0$ where $m$ is the $i$-th row of $M_k$ and $\mu_0$ is the first coordinate of $m^{[\chi]}$ (see \cite{PBB:Bou20}). The matrix $M_k^{[\chi]}$ is used in the algorithm for calculating the weight distribution of a linear code with characteristic vector $\chi$, presented in \cite{PBB:Bou20}.


\begin{definition}
Let $v\in\mathbb{Z}^{\theta}$ be a vector of length $\theta(q,k)$ with integer coordinates. For any row-vector $m$ in the matrix $M_k$, we define the vector
\[m^{[v]_r}=(\mu_0-\mu_1,\ldots,\mu_0-\mu_{q-1}),\]
where $\mu_0,\mu_1,\ldots,\mu_{q-1}$ are the coordinates of $m^{[v]}$. The matrix $M_k^{[v]_r}$ consists of the vectors $m^{[v]_r}$ as rows. The sum of the columns of $M_k^{[v]_r}$ is called the \textit{reduced distribution} of $v$ and denoted by $r(v)$.
\end{definition}

Note that $M_k^{[v]_r}$ is a $\theta(q,k)\times (q-1)$ matrix, so it has $q^k-1$ entries. Obviously,
\[\left(m^{[v]_r}\right)^T=\left(\begin{array}{rrrrrr}1&-1&0&\ldots&0&0\\1&0&-1&\ldots&0&0\\ \vdots&\vdots&\vdots&\ddots&\vdots&\vdots\\
1&0&0&\ldots&-1&0\\1&0&0&\ldots&0&-1\end{array}\right)\cdot\left(m^{[v]}\right)^T\]

\begin{lemma}\label{PBB:lem:reduced}
The reduced distribution $r(v)$ of a vector $v\in\mathbb{Z}^{\theta}$ is equal to $[(q-1)J-q{\cal N}(M_k)]v^T$
where $J$ is the $\theta\times\theta$ all 1's matrix.
\end{lemma}

\begin{proof}
Let $n=\sum_{j=1}^{\theta}v_j$. If $m$ is the $i$-th row of the matrix $M_k$, then the $i$-th coordinate of
${\cal N}(M_k)v^T$ is equal to
$$\sum_{\substack{1\le j\le\theta\\ m_j\neq 0}} v_j=n-\mu_0,$$ where $\mu_0$ is the first element of $m^{[v]}$.
Let $c_0$ be the first column of $M_k^{[v]}$. Since $n$ is the sum of the coordinates of $v$, we have
${\cal N}(M_k)v^T=Jv^T-c_0$ and so $c_0=Jv^T-{\cal N}(M_k)v^T$.

From the other hand, 
for the sum of the coordinates of $m^{[v]_r}$ we have
\[(q-1)\mu_0-\sum_{u=1}^{q-1}\mu_u=q\mu_0-\sum_{u=0}^{q-1}\mu_u=q\mu_0-\sum_{j=1}^{\theta(q,k)}v_j=q\mu_0-n.\]
Therefore the reduced distribution of $v$ is equal to

$r(v)=qc_0-Jv^T=qJv^T-q{\cal N}(M_k)v^T-Jv^T=
\left[(q-1)J-q{\cal N}(M_k)\right]v^T.
$
\end{proof}

\begin{remark} For given positive integers $k$ and $q$ ($q$ is a prime power), the transformation $r$ is defined for any integer valued vector $v$ of length $\theta(q,k)$. As we will show below, the reduced distribution can be used both for computing the weight distribution of a linear code (see \cite{PBB:Bou20}), and for calculating the covering radius.
\end{remark}

\section{Covering radius of a linear code over a prime field}
\label{PBB:sec:covering}

In this section we consider only prime fields, so we set $q$ to be a prime, and $\mathbb{F}_q=\mathbb{Z}_q=\{0,1,\ldots,q-1\}$. Theorem \ref{PBB:tm:cover_Karpovsky} gives a connection between the Vilenkin-Christenson transform and the covering radius of a linear codes.

\subsection{An algorithm using Vilenkin-Chrestenson transform}

To prove the main result (Theorem \ref{PBB:tm:cover_Karpovsky}), we need the following lemma.

 \begin{lemma}\label{PBB:lem:matrix}
 The following equality holds for any $x\in\mathbb{F}_q^s$:
\[
\sum_{\omega\in\mathbb{F}_q^s}v_{\omega}(x)=\left\{\begin{array}{l@{\quad}l}q^s,&\mbox{if} \ x=\mathbf{0},\\
0,&\mbox{if} \ x\neq\mathbf{0}.\end{array}\right.
\]
 \end{lemma}

 \begin{proof}
 This lemma follows from \cite[Chapter 5, Lemma 9]{PBB:MWS}, but we present here a proof that uses our notations.

 Since $v_{\omega}(\mathbf{0})=1$, $\sum_{\omega\in\mathbb{F}_q^{s}}v_{\omega}(\mathbf{0})=q^{s}$.

 In the case $x\neq\mathbf{0}$ we use induction by $s$. In the base step $s=1$ we have
\[
\sum_{\omega\in\mathbb{F}_q}v_{\omega}(x)=\sum_{\omega=0}^{q-1}\xi^{x\omega}=\sum_{u=0}^{q-1}\xi^{u}=
\frac{1-\xi^{q}}{1-\xi}=0, \;\; \forall x\in\mathbb{F}_q\setminus\{ 0\}.
\]

Suppose that the equality holds for some natural number $s$. Now consider the vectors in $\mathbb{F}_q^{s+1}$ and the matrix $V_{s+1}$. Let $x=(u,x')$, where $u\in\mathbb{F}_q$, $x'\in\mathbb{F}_q^s$, $x\neq\mathbf{0}$.


If $x'=\mathbf{0}$, but $u\neq 0$,
then
\[
  \sum_{\omega\in\mathbb{F}_q^{s+1}}v_{\omega}(x)  =\sum_{i=0}^{q-1}(\xi^{ui}\sum_{\omega\in\mathbb{F}_q^s}v_{\omega}(\mathbf{0}))
    =q^s\sum_{i=0}^{q-1}\xi^{i} =q^s\frac{1-\xi^{q}}{1-\xi}=0.
\]

If $x'\neq\mathbf{0}$, then
\[
  \sum_{\omega\in\mathbb{F}_q^{s+1}}v_{\omega}(x) =\sum_{i=0}^{q-1}(\xi^{ui}\sum_{\omega\in\mathbb{F}_q^s}v_{\omega}(x'))=0.
\]

Since both the base case and the inductive step have been proved as true, by mathematical induction the statement holds for every natural number $s$.
 \end{proof}

Let $C$ be a linear $[n,k]_q$ code with a parity check matrix $H$. Consider the characteristic function of the matrix $H$, defined by
\begin{equation}\label{PBB:eq:h}
h_H(x)=\left\{\begin{array}{l@{\quad}l}1,&\mbox{if\ } x \mbox{\ is a column of\ } \widehat{H},\\
0,& \mbox{otherwise,}\end{array}\right.
\end{equation}
where $\widehat{H}=(H|\alpha_2H|\ldots|\alpha_{q-1}H)$.
We use this characteristic function to compute the covering radius of the code. The following theorem holds for primes $q\ge 3$. A similar result is presented in \cite[Theorem 2]{PBB:Kar81} for the case $q=2$, but Karpovsky considers more functions.

\begin{theorem}\label{PBB:tm:cover_Karpovsky}
Let $C$ be an $[n,k]_q$-code with a parity check matrix $H$, where $q$ is an odd prime,
and $\widehat{h}:\mathbb{F}_q^{n-k}\to \mathbb{C}$ be the Vilenkin-Chrestenson transform of the characteristic function $h=h_H$.
Then the covering radius $\rho(C)$ is equal to the smallest natural number $j$ such that $\widehat{\widehat{h}^j}(y)\neq 0$ for all $y\in\mathbb{F}_q^{n-k}$, $y\neq\mathbf{0}$.
\end{theorem}

\begin{proof}
Consider the powers of $\widehat{h}(\omega)$, $\omega\in\mathbb{F}_q^{n-k}$.
\begin{eqnarray*}
\left(\widehat{h}(\omega)\right)^j&=&\left(\sum_{x\in\mathbb{F}_q^{n-k}}h(x)v_{\omega}(x)\right)^j\\
&=&\sum_{x_1,\dots,x_j\in\mathbb{F}_q^{n-k}}h(x_1)\ldots h(x_j)v_{\omega}(x_1)\dots v_{\omega}(x_j)\\
&=&\sum_{x_1,\dots,x_j\in\mathbb{F}_q^{n-k}}h(x_1)\ldots h(x_j)v_{\omega}(x_1+ \ldots + x_j).
\end{eqnarray*}
After applying the Vilenkin-Chrestenson transform on the function $\widehat{h}^j$, we have for $y\in\mathbb{F}_q^{n-k}\setminus\{\mathbf{0}\}$
\begin{eqnarray*}
\widehat{\widehat{h}^j}(y)&=&\sum_{\omega\in\mathbb{F}_q^{n-k}}\left(\widehat{h}(\omega)\right)^j v_{y}(\omega)
=\sum_{\omega\in\mathbb{F}_q^{n-k}}\left(\widehat{h}(\omega)\right)^j v_{\omega}(y)
\\
&=&\sum_{\omega\in\mathbb{F}_q^{n-k}}\;v_{\omega}(y)\sum_{x_1,\dots,x_j\in\mathbb{F}_q^{n-k}}h(x_1)\ldots h(x_j)v_{\omega}(x_1+ \ldots+ x_j)\\
&=&\sum_{x_1,\dots,x_j\in\mathbb{F}_q^{n-k}}\;\sum_{\omega\in\mathbb{F}_q^{n-k}}h(x_1)\ldots h(x_j)v_{\omega}(x_1+ \ldots+ x_j+y)\\
&=&\sum_{x_1,\dots,x_j\in\mathbb{F}_q^{n-k}}h(x_1)\ldots h(x_j)\sum_{\omega\in\mathbb{F}_q^{n-k}}v_{\omega}(x_1+ \ldots+ x_j+y).
\end{eqnarray*}
According to Lemma \ref{PBB:lem:matrix}, $\sum_{\omega\in\mathbb{F}_q^{n-k}}v_{\omega}(x_1+ \ldots+ x_j+y)\neq 0$ only if $x_1+ \ldots+ x_j+y=0$.
It turns out that $\widehat{\widehat{h}^j}(y)\neq 0$ if and only if there is at least one tuple $(x_1,\dots,x_j)$ of vectors in $\mathbb{F}_q^{n-k}$ (equal vectors are allowed) such that $h(x_u)\neq 0$ for all $x_u$, $u=1,\ldots,j$, and  $x_1+ \cdots+ x_j=-y$.
In other words, $\widehat{\widehat{h}^j}(y)\neq 0$ if and only if there is a tuple $(x_1,\dots,x_j)$ of columns in the matrix $\widehat{H}$ such that  $x_1+ \cdots+ x_j=-y$. Let $l_y$ be the number of the tuples $(x_1,\dots,x_j)$ of columns in $\widehat{H}$ whose sum is equal to $-y$. Then
$\widehat{\widehat{h}^j}(y)=l_yq^{n-k}$, and $\widehat{\widehat{h}^j}(y)\neq 0$ if and only if $y$ can be represented as a sum of $j$ columns of $\widehat{H}$. The sum $x_1+ \cdots+ x_j$ is a linear combination of at most $j$ of the columns of the parity check matrix $H$, hence we can reformulate the above statement in the following way:
  $\widehat{\widehat{h}^j}(y)\neq 0$ if and only if $y$ is a linear combination of at most $j$ columns of $H$.

It is not difficult to see that if $y$ can be represented as a sum of $j$ columns of $\widehat{H}$ (not necessarily different) then the same vector can be represented as a sum of $j+1$ columns (this conclusion is not valid if $q=2$). Indeed,
if $y=x_1+ \cdots+ x_j$ then $y=x_1+ \cdots+ x_{j-1}-\frac{q-1}{2}x_j-\frac{q-1}{2}x_j$. Therefore, if $\widehat{\widehat{h}^j}(y)\neq 0$ then $\widehat{\widehat{h}^{j+1}}(y)\neq 0$.

If $j<\rho(C)$, then there is a vector $y\in\mathbb{F}_q^{n-k}\setminus\{\mathbf{0}\}$ which is not a linear combination of $j$ columns of $H$ and then $\widehat{\widehat{h}^j}(y)= 0$. In the other hand, if $j\ge\rho(C)$ then any vector $y\in\mathbb{F}_q^{n-k}\setminus\{\mathbf{0}\}$ is a linear combination of at most $j$ columns of $H$ and therefore $\widehat{\widehat{h}^j}(y)\neq 0$ for all $y\in\mathbb{F}_q^{n-k}\setminus\{\mathbf{0}\}$.
\end{proof}

\begin{remark} If $x_1,\ldots,x_j$ are columns in $\widehat{H}$ and $x_1+ \cdots+ x_j=y$ then there exists a vector $w\in \mathbb{F}_q^n$ with $\mbox{wt}(w)\leq j$ such that $y=Hw^T$. Furthermore, it follows that $y$ is the syndrome of the coset $w+C$ and the weight of this coset is at most $j$.
\end{remark}

\begin{remark} The same algorithm can be used for computing the weight distribution of the coset leaders of a linear code over $\mathbb{F}_q$ for an odd prime $q$. If $j\ge 2$ is an  integer, then the number of the coset leaders of weight $j$ is equal to the number of the vectors $y\in\mathbb{F}_q^{n-k}\setminus\{\mathbf{0}\}$ such that $\widehat{\widehat{h}^j}(y)\neq 0$ but $\widehat{\widehat{h}^{j-1}}(y)=0$. The number of the coset leaders of weight 1 is equal to the number of the nonzero vectors $y\in\mathbb{F}_q^{n-k}$ such that $h(y)\neq 0$.
\end{remark}

\subsection{Additional Properties}

We propose some improvements in the computations presented in Theorem~\ref{PBB:tm:cover_Karpovsky}.  Note that the characteristic function $h$ of the matrix $H$ takes only integer values ($0$ and $1$). We consider the case when $q$ is an odd prime and $\mathbb{F}_q=\mathbb{Z}_q$.

Next, we use some properties of the proportionality to reduce the addends in the sum in the Vilenkin-Chrestenson transform of an integer valued function $h:\mathbb{F}_q^s\to\mathbb{Z}$ satisfying the following property: $h(x)=h(ux)$ for all $u\in\mathbb{F}_q\backslash \{0\}$ and $x \in \mathbb{F}_q^s$. Proportionality is an equivalence relation in $\mathbb{F}_q^s$ that partitions the considered set into $\theta+1$ classes, where $\theta=\theta(q,s)$. Only $\{\mathbf 0\}$ contains one element, each of all other classes consists of $q-1$ elements.

Let $e_1,\ldots,e_{\theta}$ be the vectors, corresponding to the columns of the generator matrix $G_s$ of the simplex code as it is defined in \eqref{PBB:Gk}. Note that the elements of the $\theta\times\theta$ matrix $M_s$ are the inner products $e_i\cdot e_j\in\mathbb{F}_q$. Then
\begin{equation}\label{PBB:eq0}
\widehat{h}(\mathbf{0})=\sum_{x\in\mathbb{F}_q^s} h(x) v_{\mathbf{0}}(x)=\sum_{x\in\mathbb{F}_q^s} h(x)= h(\mathbf{0})+(q-1)\sum_{i=1}^\theta h(e_i)
\end{equation}
and
$$\widehat{h}(e_i)=\sum_{x\in\mathbb{F}_q^s} h(x) v_{e_i}(x)=h(\mathbf{0})+\sum_{j=1}^\theta\sum_{u=1}^{q-1}h(e_j)v_{e_i}(ue_j)$$
\begin{equation}\label{PBB:eq_neq0}
= h(\mathbf{0})+\sum_{j=1}^\theta h(e_j)\sum_{u=1}^{q-1}(\xi^{e_i\cdot e_j})^u.
\end{equation}

\begin{lemma}
Let $h:\mathbb{Z}_q^s\to \mathbb{Z}$ be a function with the property $h(x)=h(ux)$ for all $u\in\mathbb{Z}_q\backslash \{0\}$ and $x \in \mathbb{Z}_q^s$.
If $\widehat{h}:\mathbb{Z}_q^s\to \mathbb{C}$ is the Vilenkin-Chrestenson transform of $h$ then $\widehat{h}$ is actually an integer valued function and
$\widehat{h}(\omega)=\widehat{h}(u\omega)$ for all $u\in\mathbb{Z}_q\backslash \{0\}$ and $\omega \in \mathbb{Z}_q^s$.
\end{lemma}

\begin{proof}
If $u\in\mathbb{Z}_q$, $u\neq 0$, then
$$\widehat{h}(u\omega)=\sum_{x\in\mathbb{Z}_q^s}h(x)v_{u\omega}(x)=\sum_{x\in\mathbb{Z}_q^s}h(x)v_{\omega}(ux)=\sum_{x\in\mathbb{Z}_q^s}h(ux)v_{\omega}(ux)=\widehat{h}(\omega).$$
The last equality holds because if $x$ traverses the set $\mathbb{Z}_q^s$, the same goes for $ux$ for a fixed $u\neq 0$.

To prove that $\widehat{h}$ is an integer valued function, we use \eqref{PBB:eq0} and \eqref{PBB:eq_neq0}.  Obviously, $\widehat{h}(\mathbf{0})\in\mathbb{Z}$. To prove the same for $\widehat{h}(e_i)$, we use that
\begin{equation}\label{PBB:eq_sum}
\sum_{u=1}^{q-1}(\xi^{e_i\cdot e_j})^u=\left\{\begin{array}{rl}
q-1,&\mathrm{if} \ e_i\cdot e_j=0,\\
-1,&\mathrm{if} \ e_i\cdot e_j\neq 0.
\end{array}\right.
\end{equation}
Hence $h(\mathbf{0})$, $h(e_j)$ and $\sum_{u=1}^{q-1}(\xi^{e_i\cdot e_j})^u$ in \eqref{PBB:eq_neq0} are integers, so the values of $\widehat{h}$ are integers.
\end{proof}

If we take $h$ to be the characteristic function of the linear code $C$ with parity check matrix $H$, defined in \eqref{PBB:eq:h}, we obtain the following corollary.

\begin{corollary}\label{PBB:cor1}
Let $C$ be an $[n,k]_q$-code with a parity check matrix $H$, where $q$ is an odd prime,
and $\widehat{h}:\mathbb{F}_q^{n-k}\to \mathbb{C}$ be the Vilenkin-Chrestenson transform of the characteristic function $h=h_H$.
Then the covering radius $\rho(C)$ is equal to the smallest natural number $j$ such that $\widehat{\widehat{h}^j}(e_i)\neq 0$ for all $i=1,\ldots,\theta(q,n-k)$.
\end{corollary}

\begin{proof}
Obviously, $h(ux)=h(x)$ for $u\in\mathbb{F}_q\setminus\{0\}$, $x\in\mathbb{F}_q^{n-k}$. Hence $\widehat{h}(ux)=\widehat{h}(x)$ and $\widehat{\widehat{h}^j}(ux)=\widehat{\widehat{h}^j}(x)$, for $j\in\mathbb{Z}$, $j\ge 1$.
\end{proof}

It turns out that it is enough to compute $\widehat{\widehat{h}^j}(\mathbf{0})$ and $\widehat{\widehat{h}^j}(e_i)$, $i=1,\ldots,\theta(q,n-k)$.
Using \eqref{PBB:eq_sum} we obtain that
\[\widehat{h}(ue_i)=\widehat{h}(e_i)=h(\mathbf{0})+\sum_{j=1}^{\theta(q,s)}r_{ij}h(e_j)\]
where
{\small
\[r_{ij}=\left\{\begin{array}{rl}
q-1,&\mathrm{if} \ e_i\cdot e_j=0,\\
-1,&\mathrm{if} \ e_i\cdot e_j\neq 0.
\end{array}\right.\]
In other words, if $R$ is $\theta(q,s)\times\theta (q,s)$ matrix $(r_{ij})$ then
\[\left(\begin{array}{l}\widehat{h}(\mathbf{0})\\ \widehat{h}(e_1)\\ \vdots \\ \widehat{h}(e_{\theta})\end{array}\right)=
\left(\begin{array}{l|lll}
1&(q-1)&\ldots&(q-1)\\
\hline
1&\\
\vdots&&R\\
1&
\end{array}
\right)
\left(\begin{array}{l}h(\mathbf{0})\\ h(e_1)\\ \vdots \\ h(e_{\theta})\end{array}\right)
=
\left(\begin{array}{c}
h(\mathbf{0})+(q-1)\sum_{i=1}^{\theta}h(e_i)\\
h(\mathbf{0})\left(\begin{array}{c}
1\\
\vdots\\
1
\end{array}
\right)+R\left(\begin{array}{c}
h(e_1)\\
\vdots\\
h(e_{\theta})
\end{array}
\right)\end{array}
\right)\]}

The matrix $R$ can be obtained from the matrix $M_s$ by replacing all nonzero elements by $-1$ and all zero elements by $(q-1)$. One can see that $R=(q-1)J-q{\cal N}(M_s)$ where $J$ is the $\theta(q,s)\times\theta(q,s)$ all 1's matrix.
This means that $R$ is the transform matrix of the reduced distribution of the vector $(h(e_1),\ldots,h(e_{\theta}))$ (see Lemma \ref{PBB:lem:reduced}) and we can apply the same calculation technique as in \cite{PBB:Bou20}. If $v_h=(h(e_1),\ldots,h(e_\theta))$ then
\begin{equation}\label{PBB:eqtt}
\left(\begin{array}{l}\widehat{h}(\mathbf{0})\\ \widehat{h}(e_1)\\ \vdots \\ \widehat{h}(e_{\theta})\end{array}\right)=
\left(\begin{array}{c}
\widehat{h}(\mathbf{0})\\
h(\mathbf{0})\mathbf{1}^T+Rv_h^T
\end{array}
\right)=
\left(\begin{array}{c}\widehat{h}(\mathbf{0})\\
h(\mathbf{0})\mathbf{1}^T+r(v_h)\end{array}\right).
\end{equation}

The algorithms, described in \cite{PBB:Bou20}, are related to butterfly networks and diagrams and
have very efficient natural implementations with SIMD model of parallelization
especially with the CUDA platform. They are used for computing the weight distribution of a linear code represented by its characteristic vector with respect to a generator matrix. The complexity is $O(kq^k)$ for a prime $q$.

We end this section with two examples. The first example illustrates the improvements in the calculations of the Vilenkin-Chrestenson transform. The second example gives an application of the proposed method for calculating the covering radius of a ternary linear code.

\begin{example}
For $q=3$ and $s=2$, the function $h:\mathbb{F}_3^2\to \mathbb{Z}$ is defined as follows:

\begin{center}
\begin{tabular}{|r|rrrrrrrrr|}
\hline
&0&0&0&1&1&1&2&2&2\\
$x^T$&0&1&2&0&1&2&0&1&2\\
\hline
$h(x)$&$a$&$b$&$b$&$c$&$d$&$e$&$c$&$e$&$d$\\
\hline
\end{tabular}
\end{center}

Then we have $TT(\widehat{h})=V_2\cdot TT(h)$, namely
\[TT(\widehat{h})=\left(\begin{array}{lllllllll}
1&1&1&1&1&1&1&1&1\\
1&\xi&\xi^2&1&\xi&\xi^2&1&\xi&\xi^2\\
1&\xi^2&\xi&1&\xi^2&\xi&1&\xi^2&\xi\\
1&1&1&\xi&\xi&\xi&\xi^2&\xi^2&\xi^2\\
1&\xi&\xi^2&\xi&\xi^2&1&\xi^2&1&\xi\\
1&\xi^2&\xi&\xi&1&\xi^2&\xi^2&\xi&1\\
1&1&1&\xi^2&\xi^2&\xi^2&\xi&\xi&\xi\\
1&\xi&\xi^2&\xi^2&1&\xi&\xi&\xi^2&1\\
1&\xi^2&\xi&\xi^2&\xi&1&\xi&1&\xi^2
\end{array}\right)\cdot\left(\begin{array}{c}a\\b\\b\\c\\d\\e\\c\\e\\d\end{array}\right)
=
\left(\begin{array}{rrrrr}
a&+2b&+2c&+2d&+2e\\
a&-b&+2c&-d&-e\\
a&-b&+2c&-d&-e\\
a&+2b&-c&-d&-e\\
a&-b&-c&-d&+2e\\
a&-b&-c&+2d&-e\\
a&+2b&-c&-d&-e\\
a&-b&-c&+2d&-e\\
a&-b&-c&-d&+2e\\
\end{array}\right)
.\]

Next we have
\[\mbox{
\begin{tabular}{|r|r|rrrrr|}
$x/\omega$&$h(x)$&\multicolumn{5}{c|}{$\widehat{h}(\omega)$}\\
\hline
$00$&$a$&$a$&$+2b$&$+2d$&$+2e$&$+2c$\\
$01$&$b$&$a$&$-b$&$-d$&$-e$&$+2c$\\
$11$&$d$&$a$&$-b$&$-d$&$+2e$&$-c$\\
$21$&$e$&$a$&$-b$&$+2d$&$-e$&$-c$\\
$10$&$c$&$a$&$+2b$&$-d$&$-e$&$-c$\\
\hline
\end{tabular}
}
\]
So the transform matrix is
\[\left(\begin{array}{rrrrr}1&2&2&2&2\\1&-1&-1&-1&2\\1&-1&-1&2&-1\\1&-1&2&-1&-1\\1&2&-1&-1&-1\end{array}\right), \mbox{\ while\ }{\cal N}(M_2)=\left(\begin{array}{rrrr}1&1&1&0\\1&1&0&1\\1&0&1&1\\0&1&1&1\end{array}\right).\]
\end{example}

\begin{example}\label{PBB:example}
Let $C$ be a linear ternary $[6,3]$ code with a parity check matrix
\[H=\left(\begin{array}{cccccc}
0&0&2&1&0&0\\
0&1&0&0&1&0\\
1&0&0&0&0&1\\
\end{array}\right)\]
We present the calculations in Table \ref{PBB:table-1}. Hence for this code $\rho(C)=3$. For the weight distribution of the coset leaders we have: 6 leaders of weight 1, 12 leaders of weight 2, and 8 leaders of weight 3.

\begin{table}
\caption{Calculations for Example \ref{PBB:example}.}
\begin{center}
{\begin{tabular}{c|c|r|r|r|r|r}
\hline\noalign{\smallskip}
$x/\omega /y$&$h(x)$&$\widehat{h}(\omega)$&$\widehat{h}^2(\omega)$&$\widehat{\widehat{h}^2}(y)$&$\widehat{h}^3(\omega)$&$\widehat{\widehat{h}^3}(y)$\\
\noalign{\smallskip}\hline\noalign{\smallskip}
0 0 0& 0& 6& 36& 162 & 216& 162\\
0 0 1& 1& 3& 9& 27& 27& 405\\
0 1 1& 0& 0& 0& 54& 0& 162\\
0 2 1& 0& 0& 0& 54& 0& 162\\
0 1 0& 1& 3& 9& 27& 27& 405\\
1 0 1& 0& 0& 0& 54& 0& 162\\
1 1 1& 0& $-3$& 9& 0& $-27$& 162\\
1 2 1& 0& $-3$& 9& 0& $-27$& 162\\
1 1 0& 0& 0& 0& 54& 0& 162\\
2 0 1& 0& 0& 0& 54& 0& 162\\
2 1 1& 0& $-3$& 9& 0& $-27$& 162\\
2 2 1& 0& $-3$& 9& 0& $-27$& 162\\
2 1 0& 0& 0& 0& 54& 0& 162\\
1 0 0& 1& 3& 9& 27& 27& 405\\
\noalign{\smallskip}\hline
\end{tabular}}
\end{center}
\label{PBB:table-1}
\end{table}
\end{example}

\section{Covering radius of linear codes over a composite field}
\label{PBB:sec:coveringlarge}


In this section we consider composite fields, so we set $q=p^l$ where $p$ is a prime, $l\ge 2$ is a positive integer, and $\mathbb{F}_p=\mathbb{Z}_p=\{0,1,\ldots,p-1\}$. Results in the previous section can be reformulated for a composite fields using a similar transform. Instead of the inner product in Vilenkin-Chrestenson transform we use the trace of the inner product of the input vectors \cite{PBB:Kar81,PBB:Assmus}.

Let $\zeta$ be a complex primitive $p$-th root of unity. Let
\begin{equation}\label{PBB:eq2}
\tau_{\omega}(x)=\zeta^{\mbox{Tr}(\omega\cdot x)}
\end{equation}
for any $\omega,x\in \mathbb{F}_q^s$, where $\mbox{Tr}$ is the trace map from $\mathbb{F}_q$ to $\mathbb{F}_p$. The $q^s\times q^s$ matrix $T_s=(\tau_{\omega}(x))$ determines the transform. We use \eqref{PBB:eq2} to define a Fourier type transform \cite{PBB:Assmus} called Trace transform.

\begin{definition}
Let $\mathbb{F}_q$ be a finite field with $q$ elements, $q=p^l$ for a prime $p$, and $\zeta$ be a primitive complex $p$-th root of unity.
The \textit{Trace transform} of the function $h:\mathbb{F}_q^s\to \mathbb{C}$ is a function $\widehat{h}:\mathbb{F}_q^s\to \mathbb{C}$ defined by
\begin{equation}
\widehat{h}(\omega)=\sum_{x\in\mathbb{F}_q^s}h(x)\tau_{\omega}(x)=\sum_{x\in\mathbb{F}_q^s}h(x)\zeta^{\mbox{Tr}(\omega\cdot x)},\quad \omega\in\mathbb{F}_q^s.
\end{equation}
\end{definition}

From the symmetry and linearity of the inner product and the trace map we have
\begin{equation}\label{PBB:elementary_tau}
 \tau_{\omega}(x)=\tau_{x}(\omega), \ \ \ \tau_{\omega}(x)\tau_{\omega}(y)=\tau_{\omega}(x+y).
 \end{equation}
 \begin{lemma}\label{PBB:lem:matrixtr}
 For $x\in\mathbb{F}_q^s$ the following equality holds
 \[
\sum_{\omega\in\mathbb{F}_q^s}\tau_{\omega}(x)=\left\{\begin{array}{l@{\quad}l}q^s,&\mbox{if} \ x=\mathbf{0},\\
0,&\mbox{if} \ x\neq\mathbf{0}.\end{array}\right.
\]
 \end{lemma}

 \begin{proof} This lemma is a modification of \cite[Chapter 5, Lemma 9]{PBB:MWS} but we give its proof for completeness. We use induction by $s$. In the base step $s=1$ we have
 \[
\sum_{\omega\in\mathbb{F}_q}\tau_{\omega}(x)=\sum_{\omega\in\mathbb{F}_q}\zeta^{Tr(x\omega)}=\left\{\begin{array}{ll}q,&\mbox{if} \ x=0,\\
0,&\mbox{if} \ x\neq 0.\end{array}\right.
\]
We will give some arguments for the case $x\neq 0$. When $\omega$ goes through $\mathbb{F}_q$ the multiplication $x\omega$ goes through all the elements of $\mathbb{F}_q$.
The trace is a linear map onto $\mathbb{F}_p$ with a kernel of $p^{l-1}$ elements. So
\[\sum_{\omega\in\mathbb{F}_q}\zeta^{\mbox{Tr}(x\omega)}=\sum_{\omega\in\mathbb{F}_q}\zeta^{\mbox{Tr}(\omega)}=p^{l-1}\sum_{\omega\in\mathbb{F}_p}\zeta^{\omega}=p^{l-1}\frac{1-\zeta^p}{1-\zeta}=0.\]

Suppose that the equality holds for some natural number $s$. Consider the vector $x=(u,x^{\prime})\in\mathbb{F}_q^{s+1}$, $u\in\mathbb{F}_q$, $x^{\prime}\in\mathbb{F}_q^s$. Because of the linearity of the trace map we have
\begin{eqnarray}
\sum_{\omega\in\mathbb{F}_q^{s+1}}\tau_{\omega}(x)&=&
\sum_{\omega\in\mathbb{F}_q^{s+1}}\zeta^{\mbox{Tr}(\omega\cdot x)}=
\sum_{i\in\mathbb{F}_q}\sum_{\omega^{\prime}\in\mathbb{F}_q^{s}}\zeta^{\mbox{Tr}(iu+\omega^{\prime}\cdot x^{\prime})}\nonumber\\
&=&\sum_{i\in\mathbb{F}_q}\sum_{\omega^{\prime}\in\mathbb{F}_q^{s}}\zeta^{\mbox{Tr}(iu)+\mbox{Tr}(\omega^{\prime}\cdot x^{\prime})}
=\sum_{i\in\mathbb{F}_q}\sum_{\omega^{\prime}\in\mathbb{F}_q^{s}}\zeta^{\mbox{Tr}(iu)}\zeta^{\mbox{Tr}(\omega^{\prime}\cdot x^{\prime})}\nonumber\\
&=&
\left(\sum_{i\in\mathbb{F}_q}\zeta^{\mbox{Tr}(iu)}\right)\left(\sum_{\omega^{\prime}\in\mathbb{F}_q^{s}}\zeta^{\mbox{Tr}(\omega^{\prime}\cdot x^{\prime})}\right)\nonumber\\
&=&\left(\sum_{i\in\mathbb{F}_q}\tau_{i}(u)\right)\left(\sum_{\omega^{\prime}\in\mathbb{F}_q^{s}}\tau_{\omega^{\prime}}(x^{\prime})\right)\label{PBB:tauProd}
\end{eqnarray}
The induction hypothesis and the base step give us that the sum above is nonzero iff $u= 0$ and $x^{\prime}=\mathbf{0}$. In this case the sum is $q.q^s=q^{s+1}$.

Since both the base case and the inductive step have been proved as true, by mathematical induction the statement holds for every natural number $s$.
 \end{proof}

The equation~\eqref{PBB:tauProd} shows that the matrices $T_s$ are connected by a Kroneker product, and $T_{s+1}=T_1\otimes T_s$.


Let $C$ be a linear $[n,k]_q$ code with a parity check matrix $H$.

\begin{theorem}\label{PBB:tm:cover_Karp_comp_odd}
Let $C$ be an $[n,k]_q$-code with a parity check matrix $H$, where $q=p^l$ for an odd prime $p$,
and $\widehat{h}:\mathbb{F}_q^{n-k}\to \mathbb{C}$ be the Trace transform of the characteristic function $h=h_H$.
Then the covering radius $\rho(C)$ is equal to the smallest natural number $j$ such that $\widehat{\widehat{h}^j}(y)\neq 0$ for all $y\in\mathbb{F}_q^{n-k}$, $y\neq\mathbf{0}$.
\end{theorem}

\begin{proof}
As in the proof of Theorem~\ref{PBB:tm:cover_Karpovsky} we obtain for $\omega,y\in\mathbb{F}_q^{n-k}$
\begin{equation}
\left(\widehat{h}(\omega)\right)^j=\sum_{x_1,\dots,x_j\in\mathbb{F}_q^{n-k}}h(x_1)\ldots h(x_j)\tau_{\omega}(x_1+ \ldots + x_j)
\end{equation}
and
\begin{equation}\label{PBB:eq:tildetilde}
\widehat{\widehat{h}^j}(y)=\sum_{x_1,\dots,x_j\in\mathbb{F}_q^{n-k}}h(x_1)\ldots h(x_j)\sum_{\omega\in\mathbb{F}_q^{n-k}}\tau_{\omega}(x_1+ \ldots+ x_j+y).
\end{equation}
According to Lemma \ref{PBB:lem:matrixtr}, $\sum_{\omega\in\mathbb{F}_q^{n-k}}\tau_{\omega}(x_1+ \ldots+ x_j+y)\neq 0$ only if $x_1+ \ldots+ x_j+y=0$.
So $\widehat{\widehat{h}^j}(y)\neq 0$ if and only if $y$ can be represented as a sum of $j$ columns (possible repeated) of $\widehat{H}=(H|\alpha_2H|\ldots|\alpha_{q-1}H)$.

It is not difficult to see that if $y$ can be represented as a sum of $j$ columns of $\widehat{H}$ (not necessarily different) then the same vector can be represented as a sum of $j+1$ columns. Indeed,
if $y=x_1+ \ldots+ x_j$ then $y=x_1+ \ldots+ x_{j-1}-\frac{p-1}{2}x_j-\frac{p-1}{2}x_j$. Note that $p>2$ is an odd prime. Therefore, if $\widehat{\widehat{h}^j}(y)\neq 0$ then $\widehat{\widehat{h}^{j+1}}(y)\neq 0$.

Hence,
 $\widehat{\widehat{h}^j}(y)\neq 0$ if and only if $y$ is a linear combination of at most $j$ columns of $H$. If $j<\rho(C)$, then there is a vector $y\in\mathbb{F}_q^{n-k}$ which is not a linear combination of $j$ columns of $H$ and then $\widehat{\widehat{h}^j}(y)= 0$. In the other hand, if $j\ge\rho(C)$ then any vector $y\in\mathbb{F}_q^{n-k}\setminus\{\mathbf{0}\}$ is a linear combination of at most $j$ columns of $H$ and therefore $\widehat{\widehat{h}^j}(y)\neq 0$ for all $y\in\mathbb{F}_q^{n-k}\setminus\{\mathbf{0}\}$.
\end{proof}

If $q$ is even and some vector $y$ is a sum of $j$ columns of $\widehat{H}$ then there is a possibility that $y$ is not a sum of $j+1$ columns of $\widehat{H}$. But the values of $\widehat{\widehat{h}^j}(y)$ are nonnegative because of \eqref{PBB:eq:tildetilde} and Lemma~\ref{PBB:lem:matrixtr}. So one can take the sum $\widehat{\widehat{h}^1}(y)+\cdots+\widehat{\widehat{h}^j}(y)$ instead of only $\widehat{\widehat{h}^j}(y)$. This idea  was used by Karpovsky for the case $q=2$ \cite[Theorem 2]{PBB:Kar81}.

\begin{theorem}\label{PBB:tm:cover_Karp_comp_even}
Let $C$ be an $[n,k]_q$-code with a parity check matrix $H$, where $q=2^l$,
and $\widehat{h}:\mathbb{F}_q^{n-k}\to \mathbb{C}$ be the Trace transform of the characteristic function $h=h_H$.
Let
\[g_j(\omega)=\sum_{i=1}^j\left(\widehat{h}(\omega)\right)^i,\;\omega\in\mathbb{F}_q^{n-k}, \ \ j=1,\ldots,n, \]
and $\widehat{g}_j:\mathbb{F}_q^{n-k}\to \mathbb{C}$ be the Trace transform of $g_j$.
Then the covering radius $\rho(C)$ is equal to the smallest natural number $j$ such that $\widehat{g}_j(y)\neq 0$ for all $y\in\mathbb{F}_q^{n-k}$, $y\neq\mathbf{0}$.
\end{theorem}

\begin{lemma}
Let $q=p^l$ for a prime $p$ and $h:\mathbb{F}_q^s\to \mathbb{Z}$ be a function with the property $h(x)=h(ux)$ for all $u\in\mathbb{F}_q\backslash \{0\}$ and $x \in \mathbb{F}_q^s$.
If $\widehat{h}:\mathbb{F}_q^s\to \mathbb{C}$ is the Trace transform of $h$ then $\widehat{h}$ is actually an integer valued function and
$\widehat{h}(\omega)=\widehat{h}(u\omega)$ for all $u\in\mathbb{F}_q\backslash \{0\}$ and $\omega \in \mathbb{F}_q^s$.
\end{lemma}

\begin{proof} If $u\in\mathbb{F}_q\backslash\{0\}$ then
\[\widehat{h}(u\omega)=\sum_{x\in\mathbb{F}_q^s}h(x)\tau_{u\omega}(x)=\sum_{x\in\mathbb{F}_q^s}h(x)\tau_{\omega}(ux)=\sum_{x\in\mathbb{F}_q^s}h(ux)\tau_{\omega}(ux)=\widehat{h}(\omega)
\]
because of the properties of the inner product over $\mathbb{F}_q$.

Let $e_1,\ldots,e_{\theta}$ be a maximal set of non-zero and non-proportional vectors in $\mathbb{F}_q^s$ where $\theta=\theta(q,s)$. Then
\begin{eqnarray}
\widehat{h}(\mathbf{0})&=&\sum_{x\in\mathbb{F}_q^s} h(x) \tau_{\mathbf{0}}(x)=\sum_{x\in\mathbb{F}_q^s} h(x)= h(\mathbf{0})+(q-1)\sum_{j=1}^\theta h(e_j),\label{PBB:eq0comp}\\
\widehat{h}(e_i)&=&\sum_{x\in\mathbb{F}_q^s} h(x) \tau_{e_i}(x)=h(\mathbf{0})+\sum_{j=1}^\theta\sum_{u\in \mathbb{F}_q\backslash\{0\}}h(ue_j)\tau_{e_i}(ue_j)\nonumber\\
&=& h(\mathbf{0})+\sum_{j=1}^\theta\sum_{u\in \mathbb{F}_q\backslash\{0\}}h(e_j)\tau_{e_i}(ue_j)\nonumber\\
&=&h(\mathbf{0})+\sum_{j=1}^\theta h(e_j)\sum_{u\in \mathbb{F}_q\backslash\{0\}}\zeta^{\mbox{Tr}(ue_i\cdot e_j)}.\label{PBB:eq_neq0comp}
\end{eqnarray}
So $\widehat{h}$ will be an integer valued function if $\sum_{u\in \mathbb{F}_q\backslash\{0\}}\zeta^{\mbox{Tr}(ue_i\cdot e_j)}$ are integers for all $i,j=1,\ldots,\theta$. Really, if $e_i\cdot e_j=0$ then $ue_i\cdot e_j=0$ for all $u\in \mathbb{F}_q\backslash\{0\}$ and the sum will be $q-1$.
If $e_i\cdot e_j\neq 0$ then $\{ue_i\cdot e_j|u\in\mathbb{F}_q,u\neq 0\}=\mathbb{F}_q\backslash\{0\}$. After applying the trace map over this set we obtain
$p^{l-1}$ values $a$ for every $a\in\mathbb{F}_p\backslash\{0\}$ and $p^{l-1}-1$ values 0. So
\[\sum_{u\in \mathbb{F}_q\backslash\{0\}}\zeta^{\mbox{Tr}(ue_i\cdot e_j)}=\sum_{u\in \mathbb{F}_q\backslash\{0\}}\zeta^{\mbox{Tr}(u)}=-1+p^{l-1}\sum_{a\in\mathbb{F}_p}\zeta^a=-1.\]
Therefore
\begin{equation}\label{PBB:eq:sum_tr}
\sum_{u\in \mathbb{F}_q\backslash\{0\}}\zeta^{\mbox{Tr}(ue_i\cdot e_j)}=\left\{\begin{array}{rl}
q-1,&\mathrm{if} \ e_i\cdot e_j=0\\
-1,&\mathrm{if} \ e_i\cdot e_j\neq 0\end{array}\right.
\end{equation}
This ends the proof.
\end{proof}

The equations \eqref{PBB:eq0comp}, \eqref{PBB:eq_neq0comp} and \eqref{PBB:eq:sum_tr} allow us to use the reduced distribution for calculating the trace transform as in \eqref{PBB:eqtt}.

\section{Conclusion}

This paper discusses the problem for computing the covering radius of a linear $[n,k]_q$ code over a finite field. An algorithm based on Vilenkin-Chrestenson transform is presented. The transform is applied on the characteristic function of a parity check matrix of the code. Corollary \ref{PBB:cor1} gives a method for computing the covering radius using the reduced distribution of a vector of length $\theta(q,n-k)$. This method is different from Karpovsky's algorithm for the binary case presented in \cite{PBB:Kar81} where the used transform have to be computed $\rho(C)$ times. In our algorithm this is not necessary if there is a lower bound for the covering radius $\rho$. Such a lower bound can be obtained using a fast heuristic algorithm \cite{PBB:Baicheva_Bouyukliev}.
Another advantage of the presented method is that the transformed vector is of length $\theta(q,n-k)$ and not $q^{n-k}$, as in Section \ref{PBB:tm:cover_Karpovsky} which makes it convenient to apply for a wide range of codes.

\section*{Acknowledgements}

This research was supported by a Bulgarian NSF contract KP-06-N32/2-2019.


\begin{thebibliography}{10}

\bibitem{PBB:KaskiOstergard}
P.~Kaski and P.~{\"O}sterg{\aa}rd, {\em Classification Algorithms for Codes and
  Designs} (Springer, 2006).

\bibitem{PBB:Betten06}
A.~Betten, M.~Braun, H.~Fripertinger, A.~Kerber, A.~Kohnert and A.~Wassermann,
  {\em Error-Correcting Linear Codes. Classification by Isometry and
  Applications} (Springer, 2006).

\bibitem{PBB:Cohen97}
G.~Cohen, I.~Honkala, S.~Litsyn and A.~Lobstein, {\em Covering Codes} (Elsevier
  Science B.V., North-Holland, 1997).

\bibitem{PBB:Kar79}
M.~G. Karpovsky, On the weight distribution of binary linear codes, {\em IEEE
  Trans. Inform. Theory} {\bf 25}, 105  (1979).

\bibitem{PBB:Kar81}
M.~G. Karpovsky, Weight distribution of translates, covering radius, and
  perfect codes correcting errors of given weights, {\em IEEE Trans. Inform.
  Theory} {\bf 27}, 462  (1981).

\bibitem{PBB:Bou20}
I.~Bouyukliev, S.~Bouyuklieva, T.~Maruta and P.~Piperkov, Characteristic vector
  and weight distribution of a linear code, {\em Cryptogr. Commun.} {\bf 13},
  263  (2021).

\bibitem{PBB:Huf03}
W.~C. Huffman and V.~Pless, {\em Fundamentals of Error-Correcting Codes}
  (Cambridge Univ. Press, 2003).

\bibitem{PBB:MWS}
F.~J. MacWilliams and N.~J.~A. Sloane, {\em The Theory of Error-Correcting
  Codes} (North-Holland Publishing Co., Amsterdam-New York-Oxford, 1977).

\bibitem{PBB:Bespalov}
M.~S. Bespalov, Discrete {C}hrestenson transform, {\em Probl. Inf. Transm.}
  {\bf 46}, 353  (2010).

\bibitem{PBB:Farkov}
Y.~A. Farkov, Discrete wavelets and the {V}ilenkin-{C}hrestenson transform,
  {\em Math. Notes} {\bf 89}, 871  (2011).

\bibitem{PBB:spectral}
M.~G. Karpovsky, R.~S. Stankovic and J.~T. Astola, {\em Spectral Logic and its
  Applications for the Design of Digital Devices} (John Wiley \& Sons Ltd,
  2008).

\bibitem{PBB:Assmus}
E.~F. Assmus and H.~F. Mattson, Coding and combinatorics, {\em SIAM Review}
  {\bf 16}, 349  (1974).

\bibitem{PBB:Baicheva_Bouyukliev}
T.~Baicheva and I.~Bouyukliev, On the least covering radius of binary linear
  codes of dimension 6, {\em Adv. Math. Commun.} {\bf 4}, 399  (2010).

\end{thebibliography}

\end{document}